\documentclass[conference]{IEEEtran}
\IEEEoverridecommandlockouts
\usepackage{cite}
\usepackage{amsmath,amssymb,amsfonts}
\usepackage{algorithmic}
\usepackage{graphicx}
\usepackage{textcomp}
\usepackage{xcolor}
\usepackage{braket}
\usepackage{amsthm}
\usepackage{mathtools}
\usepackage{soul}

\usepackage{amsmath}
\DeclareMathOperator*{\E}{\mathbb{E}}

\def\BibTeX{{\rm B\kern-.05em{\sc i\kern-.025em b}\kern-.08em
		T\kern-.1667em\lower.7ex\hbox{E}\kern-.125emX}}
\newtheorem{theorem}{Theorem}

\begin{document}	
	\title{Learning Hard Distributions with Quantum-enhanced Variational Autoencoders\\
  }
\author{\IEEEauthorblockN{Anantha S Rao}
		\IEEEauthorblockA{\textit{Department of Physics } \\
			\textit{Indian Institute of Science Education and Research Pune}\\
			Pune, India \\
			anantha.rao@students.iiserpune.ac.in}
		\and
		\IEEEauthorblockN{Dhiraj Madan}
		\IEEEauthorblockA{\textit{IBM Research} \\
			India \\
			dmadan07@in.ibm.com}
		\and
		\IEEEauthorblockN{Anupama Ray}
		\IEEEauthorblockA{\textit{IBM Research} \\
			India \\
			anupamar@in.ibm.com}
		\and
		\IEEEauthorblockN{Dhinakaran Vinayagamurthy}
		\IEEEauthorblockA{\textit{IBM Research} \\
			India \\
			dvinaya1@in.ibm.com }
		\and
		\IEEEauthorblockN{M.S.Santhanam}
		\IEEEauthorblockA{\textit{Department of Physics} \\
			\textit{Indian Institute of Science Education and Research Pune}\\
			Pune, India \\
			santh@iiserpune.ac.in}
	}
	\maketitle
	\begin{abstract}

An important task in quantum generative machine
learning is to model the probability distribution of measurements
of quantum mechanical systems. Classical generative models,
such as generative adversarial networks (GANs) and variational
autoencoders (VAEs), can learn the distributions of product
states with high fidelity, but fail or require an exponential number
of parameters to model entangled states. In this paper, we introduce
a quantum-enhanced VAE (QeVAE), a generative quantum-classical
hybrid model that uses quantum correlations to improve
the fidelity over classical VAEs, while requiring only a linear
number of parameters. We provide a closed form expression
for the output distributions of the QeVAE. We also empirically
show that the QeVAE outperforms classical models on several
classes of quantum states, such as 4-qubit and 8-qubit quantum
circuit states, haar random states, and quantum kicked rotor
states, with a more than 2x increase in fidelity for some states.
Finally, we find that the trained model outperforms the classical
model when executed on the IBMq Manila quantum computer.
Our work paves the way for new applications of quantum
generative learning algorithms and characterizing measurement
distributions of high-dimensional quantum states.

	\end{abstract}
	
\section{Introduction}
Research in quantum information science promises to enable the development of a fault-tolerant quantum computer that can perform certain tasks faster and more efficiently than any classical computer. To this end, several scientists have demonstrated that quantum algorithms can, in theory, outperform the best-known conventional algorithms when tackling specific problems and, in some situations, deliver a ‘quantum speedup’\cite{Biamonte2017, Aaronson2015}. For instance, certain quantum algorithms can take exponentially fewer resources for tasks such as factorization and eigenvalue decomposition, and quadratically fewer resources to search through unsorted databases \cite{shor1994algorithms, harrow2009quantum, grover1996fast}. This pursuit of ‘quantum speedup’ has motivated generations of physicists and engineers to discover novel algorithms that leverage the properties of superposition, entanglement and interference. 
	
Through advances in processing power and algorithmic ability of computing devices, machine learning techniques have evolved into fundamental tools for detecting patterns in data. Theoretically, models like deep neural networks have the potential to learn some of the most complex patterns that exist in nature or human-made systems and have been demonstrated to be highly competent at complex tasks like playing Go, identifying protein structures, and self-driving cars~\cite{Silver2017, Jumper2021, selfdrivingcars_2022}. However, many tasks are still intractable or very expensive for these methods. Some learning tasks, for example, include sampling from complex distributions or estimating the average values of numerous parameters under a complicated distribution, both of which are typically intractable (requiring exponential time or space resources). Moreover, certain distributions derived from quantum-mechanical systems are fundamentally intractable to conventional approaches~\cite{Bremner2010, fefferman2015power,Sweke2021quantumversus}. Enhancing and augmenting classical machine-learning methods using quantum correlations has been the focus of quantum-enhanced machine learning and our work. 
	

One of the important tasks for machine learning is that of generative learning, that involves modeling or learning a distribution given independent samples from the same. Previously, models such as Generative Adversarial networks (GANs~\cite{gans}) and Variational Autoencoders (VAEs~\cite{vae}) have been used to learn distributions of classical data such as texts and images.

 In this paper, we study the problem of modeling the measurement distributions obtained from unknown quantum states. This problem is fundamental in quantum information science, as it can reveal useful information about the properties and dynamics of quantum systems. Moreover, it can enable applications such as quantum state reconstruction, and entanglement quantification~\cite{Huang2022, Guillermo2021, Rocchetto2019, Huang2020}. These applications can help us understand and manipulate quantum systems in various fields, such as chemistry, materials science, and cybersecurity.

For our problem setting, Variational Autoencoders have also been shown to learn and compress the measurement distribution obtained from product states, termed as `easy states' but require an exponential number of parameters to learn the distribution of Haar random states (`hard states') \cite{rocchetto2018learning, niu2020learnability}. In this work, we propose a quantum-enhanced VAE (QeVAE) which incorporates quantum correlations through quantum circuits to enhance the performance of classical VAEs. We show that a QeVAE constructed by substituting the decoder (generator) of a conventional VAE with a parameterized quantum circuit can learn these `hard states' with only a linear number of parameters in the number of qubits and outperform conventional methods when the dataset contains quantum-correlations.
	
In this paper, we make three main contributions to the field of quantum generative learning. Firstly, we propose the Quantum-enhanced VAE (QeVAE), which can enhance the expressive power of classical VAEs and produce distributions that are classically intractable. Secondly, we provide a mathematical closed-form expression to theoretically analyze the class of models that QeVAE can model better than classical models. Finally, we demonstrate experimentally that our QeVAE outperforms classical VAEs on modeling measurement samples of different classes of quantum states, such as haar random states, quantum circuit states, and quantum kicked rotor states, and report an increase in fidelity by more than 2x for an 8-qubit quantum circuit state.
	
Our work is organized as follows: In Section \ref{sec:background}, we discuss the background work on generative learning, variational autoencoders, and the problem of learning measurement distributions of quantum states. In Section \ref{sec:QeVAE}, we propose the QeVAE and mathematically characterize the output distribution of the model. We describe our experiments and results in section \ref{sec:methods} and \ref{sec:results} respectively. Finally, in Section \ref{sec:Conlcusion}, we discuss our conclusions and future outlook in the greater context of developing quantum algorithms for generative learning.

\section{Background and Related work}
\label{sec:background}

\subsection{Generative learning and Variational Auto Encoders}
The task of generative modeling involves modeling a parameterized distribution $p_\theta(\mathbf{x})$, given independent and identically distributed (iid) samples from the distribution $\{\mathbf{x_i}\}_{i=1}^m$ where each $\mathbf{x_i} \sim p(\mathbf{x})$.
Here the goal is to maximize the log likelihood of the data given by $L (\theta) = \sum_i \log{p_\theta(\mathbf{x_i})}$.
 One of the common approaches for generative modeling is that of adversarial learning (for example, GANs, \cite{gans}), which involves learning a parameterized generative network and a parameterized discriminative network. While the generative network maps samples from a fixed distribution to the data distribution, the discriminative network distinguishes between data samples and the generated samples.

On the other hand, variational learning involves modeling the distribution through a latent vector $\mathbf{z}$. One can define a prior on the latent variables $\mathbf{z} \sim p(\mathbf{z})$ and a parameterized likelihood function $p_\theta(\mathbf{x} \vert \mathbf{z})$.
The joint distribution over the variables $\mathbf{x}$ and $\mathbf{z}$ can then be written as $p(\mathbf{x},\mathbf{z})= p(\mathbf{z}) p_\theta (\mathbf{x} \vert \mathbf{z})$.
The log likelihood can then be expressed as $L(\theta)=\log(p_\theta(\mathbf{x})) = \log(\sum_{\mathbf{z}} p(\mathbf{x}, \mathbf{z}))$. Variational learning considers a lower bound for the above as:-

\begin{equation}
    \log(p_{\theta}(\mathbf{x})) \geq \E_{\mathbf{z} \sim q(\mathbf{z})} \left(\frac{\log(p(\mathbf{x},\mathbf{z}))}{q(\mathbf{z})}\right)
\end{equation}

which holds for any distribution $q(\mathbf{z})$. The above lower bound called Evidence Lower Bound (ELBO) holds tight when $q(\mathbf{z}) = p (\mathbf{z} \vert \mathbf{x})$. Variational Autoencoders \cite{vae} seek to maximize the ELBO where $q_\phi(\mathbf{z} \vert \mathbf{x})$ can be defined through a parameterized neural network also called the encoder. Here we  can split the ELBO as 
\begin{equation}
\text{ELBO} = \E_{q(\mathbf{z}\vert\mathbf{x})}\left(\log p(\mathbf{x} \vert \mathbf{z})\right) -  KL (q(\mathbf{z}\vert \mathbf{x}) \vert \vert p(\mathbf{z})))
\end{equation}

\begin{figure}[h!]
    \centering
    \includegraphics[width=0.8\linewidth]{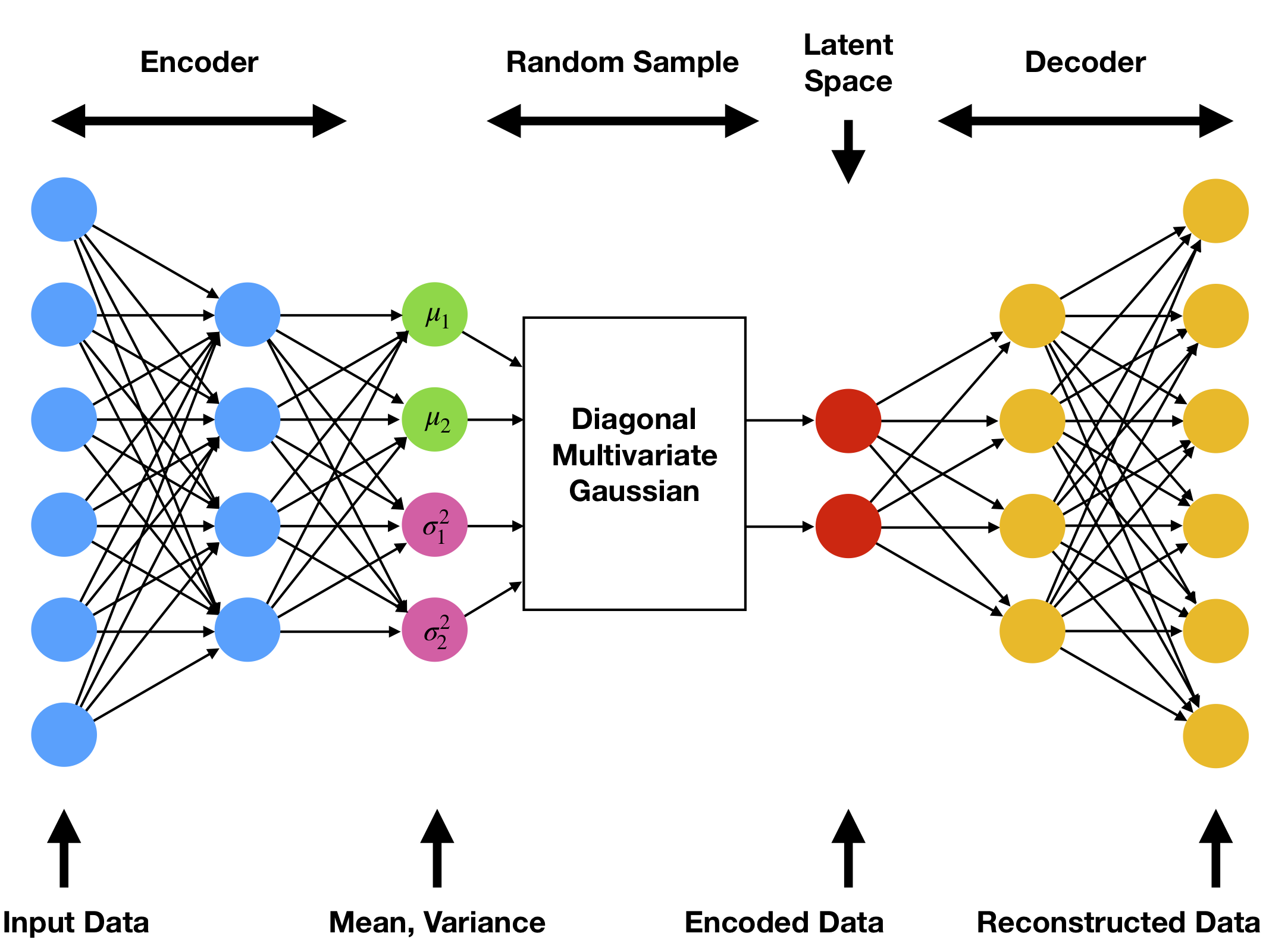}
    \caption{\textbf{A schematic for a classical VAE.} A classical VAE consists of an encoder, decoder and a continuous latent dimension that can be modeled as a multivariate gaussian with diagonal covariance. The dataset \textbf{x} is mapped to a distribution p(z) and after training random samples \textbf{x} can be generated by sampling from the decoder.}
    \label{fig:classical_vae}
\end{figure}
 The algorithm involves using the posterior network (encoder) to create the distribution $q_\phi(\mathbf{z} \vert \mathbf{x})$ and provides samples $\mathbf{z}$, which are then passed through the likelihood network (decoder) to produce a distribution over $\mathbf{x}$ conditioned on sampled vector $\mathbf{z}$. One then seeks to maximize the difference of log likelihood of generated distribution $p(\mathbf{x} \vert \mathbf{z})$ and the KL divergence between posterior ($q(\mathbf{z} \vert \mathbf{x})$) and prior distributions  $p(\mathbf{z})$. The first term (reconstruction term) tries to maximize the likelihood of recovering back $\bf{x}$ from the latent variable. The second term is a regularization term that tries to ensure that posterior distribution is close to prior.
 In order to be able to backpropagate through the sampling step, one uses a trick called the reparametrization trick. For example, when $q_\phi(\mathbf{z} \vert \mathbf{x})$ is a Gaussian with mean $\mu(\mathbf{x} ; \theta)$ and a covariance matrix (typically diagonal) $\Sigma(\mathbf{x}; \theta)$, to sample $\mathbf{z} \sim \mathcal{N}(\mathbf{\mu, \Sigma})$, one redefines $\mathbf{z} = \mu + \Sigma^{1/2} \epsilon$
where $\epsilon \sim \mathcal{N}(0, I)$. Separating out the noise $\epsilon$, enables us to backpropagate through the $\mu$ and $\Sigma$ which depend on the parameters of the network.

 Giving a higher weight to the second term leads to posterior becoming independent of $\bf{x}$ (known as the posterior collapse problem\cite{postcollapse}).
 In practice one often weighs the KL divergence term by a hyperparameter $\beta$ to control the effect of regularization term \cite{betavae}, producing the overall cost function as: 
\begin{equation}
\E_{q_{\phi}(\mathbf{z}\vert \mathbf{x})} \left(\log p_{\theta}(\mathbf{x} \vert \mathbf{z})\right) -  \beta KL (q_{\phi}(\mathbf{z}\vert \mathbf{x}) \vert \vert p(\mathbf{z})))
\end{equation}
where $\beta$ is a hyperparameter that indicates the relative importance of the regularization term with respect to reconstruction term.

\subsection{Quantum computation and quantum machine learning}
Quantum computing is a model of computation with a potential to provide a speedup over its classical counterpart. 
The fundamental units in quantum computing framework consist of  quantum bits or \textbf{qubits}. A qubit can take value as a unit vector in two dimensional complex Hilbert Space $\mathbb{C}^2$. The basis states \{$\ket{0}$ and $\ket{1}$\} correspond to the classical bits \{0,1\}. An arbitrary state $\ket{\psi}$ can be considered as a superposition of the basis states $\ket{\psi} = \alpha \ket{0} +\beta \ket{1}$, where $\vert \alpha\vert ^2+\vert \beta\vert^2=1 $.

More generally, an $n$ qubit state lies in a Hilbert space spanned by basis states corresponding to $2^n$ classical bit strings, $\ket{0...0}$ through $\ket{1...1}$.
An arbitrary state can be a unit vector spanned by the $2^n$ basis strings as $\ket{\psi} = \sum_{x \in \{0,1\}^n} \alpha_x \ket{x}$, where $\sum_x \vert \alpha_x \vert^2 =1$.

If there are $n$ qubits each in state $\{\ket{\psi_i}\}_{i=1}^m$, then the overall state of the n qubit system (product state) is $\otimes_{i=1}^n \ket{\psi_i}$.

However there are $n$ qubits states that cannot be expressed in this form and are refered to as entangled states. The measurement of these states yields correlated bit strings.

The basic operations on qubits include quantum gates, which are unitary operators acting on one or more qubits. Common examples of single qubits gates include the Hadamard, Pauli Gates(X,Y,Z), and S gate.

The exponentiated Pauli Gates also provide us a set of parameterized gates, $R_x(\theta) = exp(- i \frac{\theta}{2} X)$ , $R_y(\theta) = exp(- i \frac{\theta}{2} Y)$ and $R_z(\theta) = exp(- i \frac{\theta}{2} Z)$.
A 2 qubit gate can be described by a 4 dimensional unitaries. A common 2 qubit gate is $CNOT$ which acts as $CNOT \ket{x,y} = \ket{x, x\oplus y}$ on basis states. 
A quantum circuit consists of a sequence of single qubit and two-qubit gates acting on an initial state (typically $\ket{0...0}$) followed by measurement of the final state.
A measurement of a state $\ket{\psi} = \sum_x \alpha_x \ket{x}$ yields one of the bit strings $x$ with probabilities $\vert \alpha_x\vert ^2$.

A parameterized quantum circuit (PQC) is a quantum circuit that has some gates that depend on parameters $\theta$. The unitary operator of the PQC is denoted by $\hat{U}(\theta)$, which can produce a state $\ket{\psi(\theta)} = \hat{U} (\theta) \ket{0}^{\otimes n}$. The PQC can also be conditioned on an input $\mathbf{x}$ as $\hat{U}(\mathbf{x}, \theta)$. This can be decomposed into a feature map $\hat{U}_\phi(\mathbf{x})$ and a trainable ansatz $\hat{V}(\theta)$. 
The feature map is a data encoding circuit that transforms input data $\bar{x}\in\mathbb{R}^n$ into a quantum state using 
single qubit and two qubit gates, potentially parameterized by the input variables.
The ansatz is a parameterized circuit that consists of alternating rotation layers and entanglement layers. The rotation layers are single-qubit gates applied on all qubits. The entanglement layer uses two-qubit gates to entangle the qubits according to a predefined scheme. The parameters of the PQC can be optimized to achieve a certain goal, such as minimizing the energy of a quantum system or maximizing the accuracy of a machine learning task. The parameters can be updated iteratively using classical optimization methods such as gradient-based (ADAM~\cite{kingma2014adam}, SPSA~\cite{spall1998implementation}) or gradient-free (COBYLA) algorithms. The measurement distribution of $\ket{\psi(\mathbf{x}, \theta)} = \hat{U}(\mathbf{x},\theta) \ket{0}$ in Z-basis defines the conditional distribution $p(\mathbf{y}\vert \mathbf{x}) = \vert \braket{y| \psi(\mathbf{x},\theta)} \vert^2$  $\forall y \in \{0,1\}^n$.

A variational quantum algorithm (VQA) is then commonly defined as a hybrid quantum-classical algorithm that utilizes a PQC to optimize a cost function~\cite{Bharti2022}. The cost function can be based on the sampled measurement distribution of the state (QML, or our work), or computing the expectation value of a Hamiltonian (for quantum chemistry problems). In the context of Quantum ML, this model is also called a quantum neural network. Moreover, a quantum circuit born machine (QCBM) is a generative quantum neural network that uses a PQC to represent the probability distributions over bit strings as the measurement distribution of the quantum state $\ket{\psi(\theta)} = \hat{U}(\theta) \ket{0}^{\otimes n} $. The parameters of the model can be trained by minimizing the cross entropy loss with the generated samples. After training, the model generates samples from the desired distribution.

\subsection{Prior work in Learning distributions of quantum states}
\label{sec:prior-work}
In this section, we review prior work on learning the measurement distribution of quantum states, i.e., given samples from measuring an $n$-qubit quantum state 
\begin{align}
	\ket{\psi} = \sum_{x \in \{0,1\}^n}\alpha_x \ket{x},
\end{align}
with probabilities $p(x) = |\alpha_x|^2$ for each bitstring $x$, we want to find a model that approximates the distribution with parameters $\theta$ and yields a distribution $p_{\theta}(x)$. This problem can have potential applications in quantum state compression, quantum state transfer, and quantum state discrimination.
Previous works have shown that classical generative models such as restricted Boltzmann machines (RBMs), variational autoencoders (VAEs), and autoregressive models can model such distributions but require an exponential number of parameters to learn and generate samples~\cite{carleo2017solving, rocchetto2018learning, niu2020learnability}. Moreover, a recent work based on the probably approximately correct (PAC) framework has shown that these distributions can be efficiently learned with quantum resources, but not with purely classical approaches~\cite{hinsche2021learnability}. However, the quantum learner proposed in that work assumes the availability of a fault-tolerant quantum computer.

There have been some previous works on employing quantum algorithms and the variational autoencoder framework for different problems using classical data. The first proposed algorithm within the VAE framework utilized the annealing-based framework~\cite{khoshaman2018quantum}. A recent work focused on improving the latent space representation of classical VAEs through parameterized quantum circuits (PQCs)~\cite{9799154}. Only recently, PQCs within the VAE framework were proposed for the problem of drug discovery~\cite{li2022scalable}. To the best of our knowledge, this is the first work to consider quantum VAEs for data generation from quantum states.

We propose a quantum-enhanced variational autoencoder (QeVAE) that can learn the measurement distribution of an unknown quantum state using noisy quantum devices. Our QeVAE can reconstruct the measurement distribution through an iterative learning process that involves a parameterized quantum circuit (PQC) as the generative model and a classical neural network as the inference model. We also show that our QeVAE reduces to a quantum circuit born machine (QCBM) in the zero latentsize limit. We hypothesize that our QeVAE can leverage the quantum properties of superposition and entanglement to learn the complex and high-dimensional measurement distributions of quantum states.
\section{Quantum-enhanced Variational Autoencoders}
	\label{sec:QeVAE}
In this section, we define a quantum-enhanced variational autoencoder (QeVAE) to model the measurement distribution of an unknown $n$-qubit quantum state. As a baseline, classical VAEs have already been used to model such a distribution \cite{rocchetto2018learning}. The hybrid model that we propose consists of a feedforward classical encoder, a continuous latent space, and a parametrized quantum circuit as a decoder. We model the approximate posterior (encoder network) $Q_{\phi} (\mathbf{z} \vert \mathbf{x})$ through a classical feedforward neural network and the latent variable as $\mathbf{z} \sim \mathcal{N} (0,I)$. The likelihood (generator) distribution $p_{\theta} (\mathbf {x} \vert \mathbf{z})$ is defined via a quantum circuit i.e. $p_{\theta} (\mathbf{x} \vert \mathbf{z}) = \vert \braket{x | \hat{U}(\theta,\mathbf{z}) |0\rangle^{\otimes n}} \vert ^2$. The Evidence Lower bound loss (ELBO) is optimized through a classical optimizer such as ADAM. The model is trained to mimic the given measurement distribution of states. During training, the parameters of the encoder and the rotation gates in the decoder (with a pre-selected entanglement type) are iteratively varied and learned. Such a model has multiple applications: It will enable scientists to generate certain quantum states in different physical quantum computers just by knowing the set of rotation and entangling gates to perform. The algorithm also has applications in state compression and transferring a state from one system to another upto a phase (entire phase information can be learnt with additional models\cite{rocchetto2018learning}). We now present a theorem that allows us to mathematically characterize the class of distributions that can be obtained via the above model.
	
	\begin{theorem}
		Consider a latent variable model with $\mathbf{z} \sim p(\mathbf{z})$ and $p(\mathbf{x} \vert \mathbf{z}) = \vert \bra{\mathbf{x}}   V_\theta U_\phi(\mathbf{z})  \ket{0^n} \vert^2 $, where $U_{\phi}(\mathbf{z})$ is a feature map and $V_{\theta}$ is a parameterized ansatz.
		Then \begin {enumerate}
		\item $\exists$ a density matrix $\rho$, such that $p\mathbf(x) = \bra{x} V_{\theta} \rho V_{\theta}^{\dagger} \ket{x}$. In other words the distribution of $\mathbf{x}$ can be obtained by evolving a density matrix $\rho$ under unitary ansatz  $V_\theta$ followed by measurement in standard basis.
		\item Conversely, for each density matrix $\rho$, there exists a prior $p(\mathbf{z})$ and a feature map $U_\phi(\mathbf{z})$, such that $p\mathbf(x) = \bra{x} V_{\theta} \rho V_{\theta}^{\dagger} \ket{x}$.
	\end{enumerate}
    \label{thm1:density-matrix-vae}
\end{theorem}

\begin{proof}
	\begin{enumerate}
		\item
		
		\begin{align*}
			p(\mathbf{x}) &= \int p(\mathbf{x}, \mathbf{z}) d \mathbf{z}\\
			&=\int p(\mathbf{z}) \vert \bra{\mathbf{x}}   V_\theta U_\phi(\mathbf{z})  \ket{0^n} \vert^2   d \mathbf{z} \\
			&=\int p(\mathbf{z})  \bra{\mathbf{x}}   V_\theta U_\phi(\mathbf{z})  \ket{0^n} \bra{0^n}   U_\phi(\mathbf{z})^{\dagger} V_\theta^{\dagger}   \ket{\mathbf{x}}    d \mathbf{z} \\
			&=  \bra{\mathbf{x}}   V_\theta (\int p(\mathbf{z}) U_\phi(\mathbf{z})  \ket{0^n} \bra{0^n}   U_\phi(\mathbf{z})^{\dagger} d \mathbf{z} )V_\theta^{\dagger}   \ket{\mathbf{x}}     \\
			 &\text{ ( By linearity of inner product) }
		\end{align*}

		We define $\rho \coloneqq \int p(\mathbf{z}) U_\phi(\mathbf{z})  \ket{0^n} \bra{0^n} U_\phi(\mathbf{z})^{\dagger} d \mathbf{z}$ \\
		\textbf{Note} \begin {enumerate} \item $\rho$ is a valid density matrix i.e. $\rho^{\dagger} = \rho$,  $\rho \geq 0$ and $Tr(\rho) = I$ 
		\item $\rho$ is independent of both $\mathbf{x}$ and $\mathbf{z}$.
	\end{enumerate}
	
	Then, \begin{align*}
		p(\mathbf{x}) & = \bra{\mathbf{x}}   V_\theta \rho V_\theta^{\dagger}   \ket{\mathbf{x}} 
  \end{align*}
  
	Thus, such the $p(x)$ is equivalent to preparing a density state $\rho$, evolving under unitary $V_{\theta}$, and performing a measurement in standard basis. 
	
	\item For the second part, given a density matrix $\rho$, consider the sepctral decomposition of $\rho$ as $\rho = \sum_z \lambda_z \ket{\psi_z} \bra{\psi_z}$.
	Since $\rho$ is a valid density matrix, $\lambda_z$'s define a distribution $p(\mathbf{z}) = \lambda_z$. Now, for each $\mathbf{z}$, choose a unitary such that $U_{\phi(\mathbf{z})} \ket{0} = \ket{\psi_z}$. Now, one can verify that such a feature map satisfies the required equation.
\end{enumerate}
\end{proof}
Moreover, such a distribution defined using PQCs can encode information about the entanglement structure of the states via the entangling ansatz $V_\theta$ and density matrix $\rho$ which is inaccessible to purely classical models. Note that this fully includes the set of QCBM distributions since setting $\rho = \ket{0} \bra{0}$, will enable us to get QCBMs.


\section{Methods}
\label{sec:methods}
\begin{figure}
	\centering
	\includegraphics[width=\linewidth]{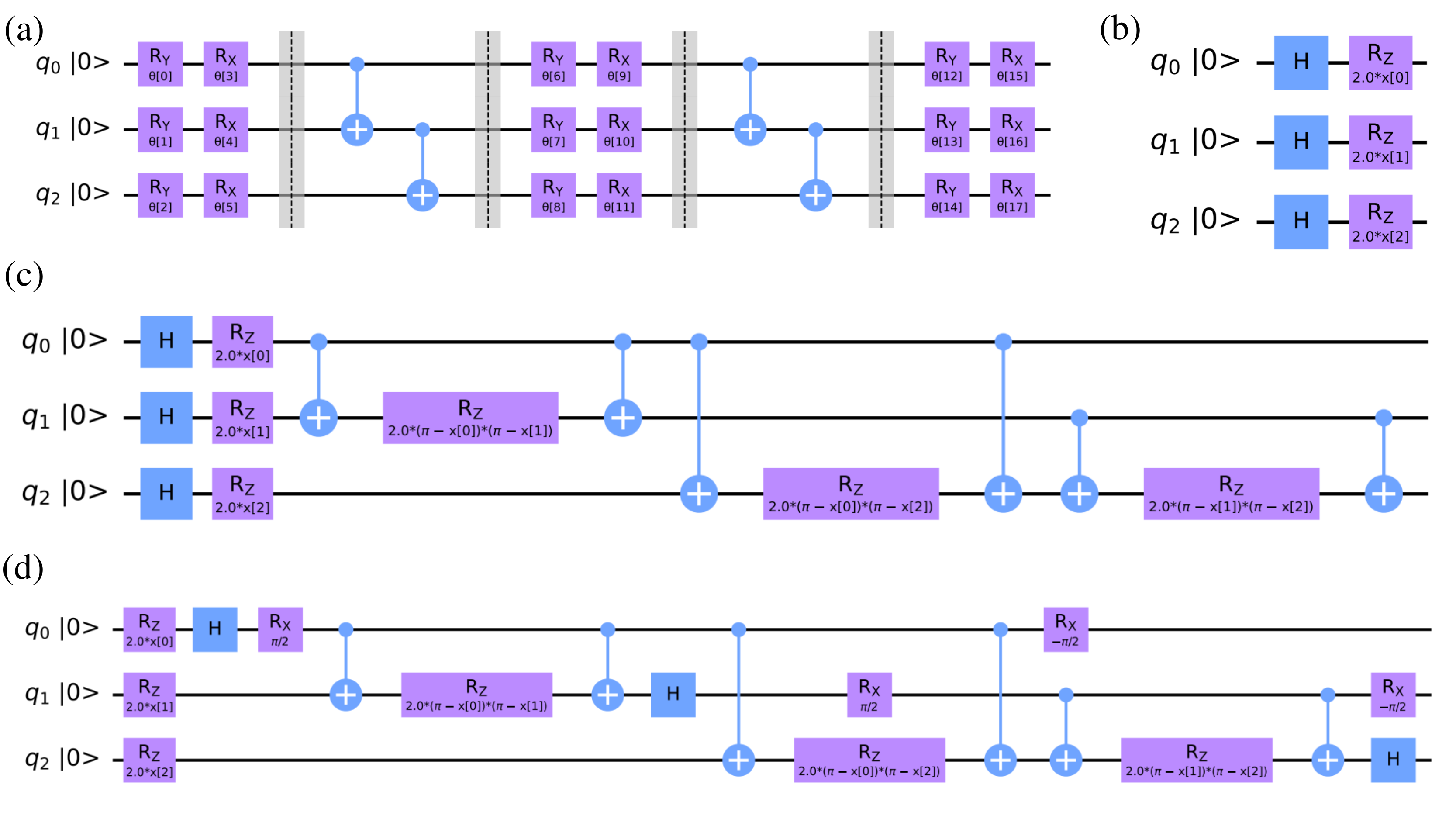}
	\caption[Ansatz and Feature maps for learning quantum distributions]{\textbf{Ansatz and Feature maps for learning quantum distributions} (a) A two-local ansatz on three qubits with two repeating layers of $Rx$ and $Ry$ gates along with linear entanglement (b) A Pauli-Z feature map that embeds a three-dimensional vector. `H' represents the Hadamard gate. (c) A Pauli-ZZ feature Map (d) A Pauli-(X, XY) feature map on a three qubit system.}
	\label{fig:ansatz-featuremaps}
\end{figure}
In this section, we provide the methods followed to setup a hybrid quantum-classical neural network that follows the variational autoencoder framework. We provide details on the architecture of the model, the training algorithm, metrics followed, datasets utilized and end the section with the hyperparameters used for various models.

\subsection{Architecture Details}

The hybrid quantum-classical neural network consists of three components: An encoder, a continuous latent-space and a decoder. In section we discuss the structure of each component.

The \emph{encoder} is modeled through a classical network which defines the mean $\mu(\mathbf{x}, \theta)$, and diagonal covariance $\Sigma(\mathbf{x}, \theta)$ for the posterior distribution $Q(\mathbf{z} \vert \mathbf{x})$, modeled as a multivariate Gaussian.
In our experiments, we use feed-forward neural networks where the input bitstring is transformed through successive operations of linear layers each followed by a non-linearity. Here, one starts with $h^{0}=x$, which is the initial input. The $l^{th}$ layer transforms this as $h^{l}= f^{l}(h^{l-1})= \phi(W^{l} h^{l-1}+b^l)$, where $\Phi$ is the Leaky ReLu activation function and $\{W^l, b^l\}$ are the parameters of the $l^{th}$ layer of the neural network (weights, biases). The final layer gives us the vectors corresponding to mean $\mu$ and the diagonal entries of log covariance matrix. Since we benchmark the performance of classical VAEs with QeVAEs, the architecture of the encoder is constant over all simulations. The encoder consists of 2 hidden layers each containing 8 and 7 neurons with the leak of the activation function set to 0.01. The output of the encoder is propagated to the latent space that is modelled as a continuous gaussian with diagonal covariance. The encoded sample is obtained from the latent space through the reparametrization trick. 

The \emph{decoder} network in the QeVAE is modeled via a variational quantum circuit with trainable rotation gates. We use a Pauli feature map (such as the Z or ZZ featuremap given in figure \ref{fig:ansatz-featuremaps}) to encode the sampled latent variable $\mathbf{z}$ onto the circuit. The Pauli featuremap is a data encoding circuit that transforms input data $\bar{x}\in\mathbb{R}^n$ $U_{\Phi(\vec{x})} \ket{0}=\exp\left(i\sum_{S\subseteq [n]} 
\phi_S(\vec{x})\prod_{i\in S} P_i\right) \ket{0}$ \cite{havlivcek2019supervised}. The variable $P_i \in \{X, Y, Z, I\}$ denotes the Pauli matrices. The index $S$ describes connectivities between different qubits in a given circuit. 
If $S=\{Z\}$, we obtain a $Z$ feature map. For $S=\{Z, ZZ\}$ we have ZZ feature map.

We then use a two-local of ansatz with linear entanglement due to their compatibility with current-era hardware. The two-local circuit is a parameterized circuit consisting of alternating rotation layers and entanglement layers. The rotation layers are single qubit gates applied on all qubits while the entanglement layer uses two-qubit gates to entangle the qubits according to a predefined scheme. In addition, while training, we find it useful to add an additional linear feedforward layer before the quantum circuit, after sampling a latent vector. This provides two benefits: (a) Flexibility in choosing a latent size. The layer can linearly transforming the latent vector of a different size to fit the input requirements of the quantum circuit; (b) This also adds power to the network by introducing additional parameters. 

To benchmark the performance of the QeVAE, we train classical VAEs on different 4 qubit and 8 qubit states, and different number of layers. The number of trainable parameters is kept fixed close to $2^{\text{No. of qubits}}$. Table \ref{tab:cvae-parameters} shows the sizes of classical decoder layers for circuits corresponding to 8 qubit measurements. 
\begin{table}[h!]
	\caption{No of neurons and number of hidden layers for different classical VAE benchmarks for 8 qubit states}
	\begin{tabular}{|c|c|l|}
	\hline
	\#Hidden layers & \# parameters & \#Neurons in each layer \\ \hline
	1 & 256 & 8,15,8 \\ 
	2 & 269 & 8,10,9,8 \\ 
	3 & 255 & 8,6,7,9,8 \\ 
	4 & 266 & 8, 7, 6, 6, 7, 8 \\ 
	5 & 258 & 8, 7, 5, 4, 5, 7, 8 \\ \hline
	\end{tabular}
	\label{tab:cvae-parameters}
\end{table}
\subsection{Training details}
We train both the classical and hybrid models using ADAM optimizer. The training is iterative where each iteration involves forward propagating a single bitstring (`x') (obtained from the distribution) through the encoder, obtaining a sample from the latent space (`z'), preprocessing the sample through a linear layer, embedding onto the PQC, propagating through the quantum circuit, and computing the measurement distribution with multiple shots. The loss function utilizes the quasi-probability of $p(x|z)$ and the KL divergence of the latent space. To avoid overfitting on the training data, we perform early stopping on the validation set loss with a patience factor $\delta$. We use python packages Pytorch and the Torchconnector module in qiskit to build these hybrid models\cite{paszke2019pytorch, Qiskit}. In table \ref{tab:hyperparameters}, we summarize the various hyperparameters used in the QeVAE model. 
\begin{table}[h!]
	\caption{}
	\begin{tabular}{ |l|l|l|}
		\hline
		No & Hyperparameter & Value \\
		\hline
		1 & Feature Map & Z - ZZ Pauli map \\
		2 & Ansatz & Two-local \\
		3 & Entanglement & Linear \\
		4 & Measurement basis & Z \\
		5 & Encoder learning rate & 0.001 - 0.01 \\
		6 & Decoder learning rate & 0.001 - 0.009 \\
		7 & No of latent variables & 0 - number of qubits \\
		8 & Batchsize & 16-64 \\
		9 & KL-term weight ($\beta$) & 0.5-2 \\
		10 & Training type & Annealing, Fixed, Stepfunction \\
		11 & Patience factor ($\delta$) & 5-7 \\ \hline
	\end{tabular}
	\label{tab:hyperparameters}
\end{table}

\subsection{Metrics}
The metric we use to quantify the generated distribution is the fidelity between two discrete distributions. If $\rho$ and $\sigma$ are $n$-qubit states, we say that $\sigma$ is a good representation of $\rho$ if the fidelity F = Tr($\sqrt{\rho^{1/2}\sigma\rho^{1/2}}$) $> 1 - \epsilon$ for an $\epsilon >$ 0. Through the result of \cite{fuchs1994ensemble}, the fidelity can be expressed in terms of the probability distributions over a measurement that maximally distinguishes the two states. Thus given two random variables X, Y with probabilities $p = (p_1 , p_2 , ..p_n )$ and $q = (q_1 , q_2 , ..q_n )$, the fidelity of X and Y is defined to be the quantity:
\begin{equation}
F(X,Y) = (\sum_i \sqrt{p_i q_i} )^2
\end{equation}

where the measure $\sum_i \sqrt{p_i q_i}$  is the Bhattacharyya coefficient between the two distributions. To compute the fidelity between the original and learnt distributions, we propagate 5000 random samples from the latent space through the decoder and construct the output distribution by computing the average distribution. 
\subsection{Datasets}
\begin{figure*}[h!]
\centering
\includegraphics[width=\textwidth]{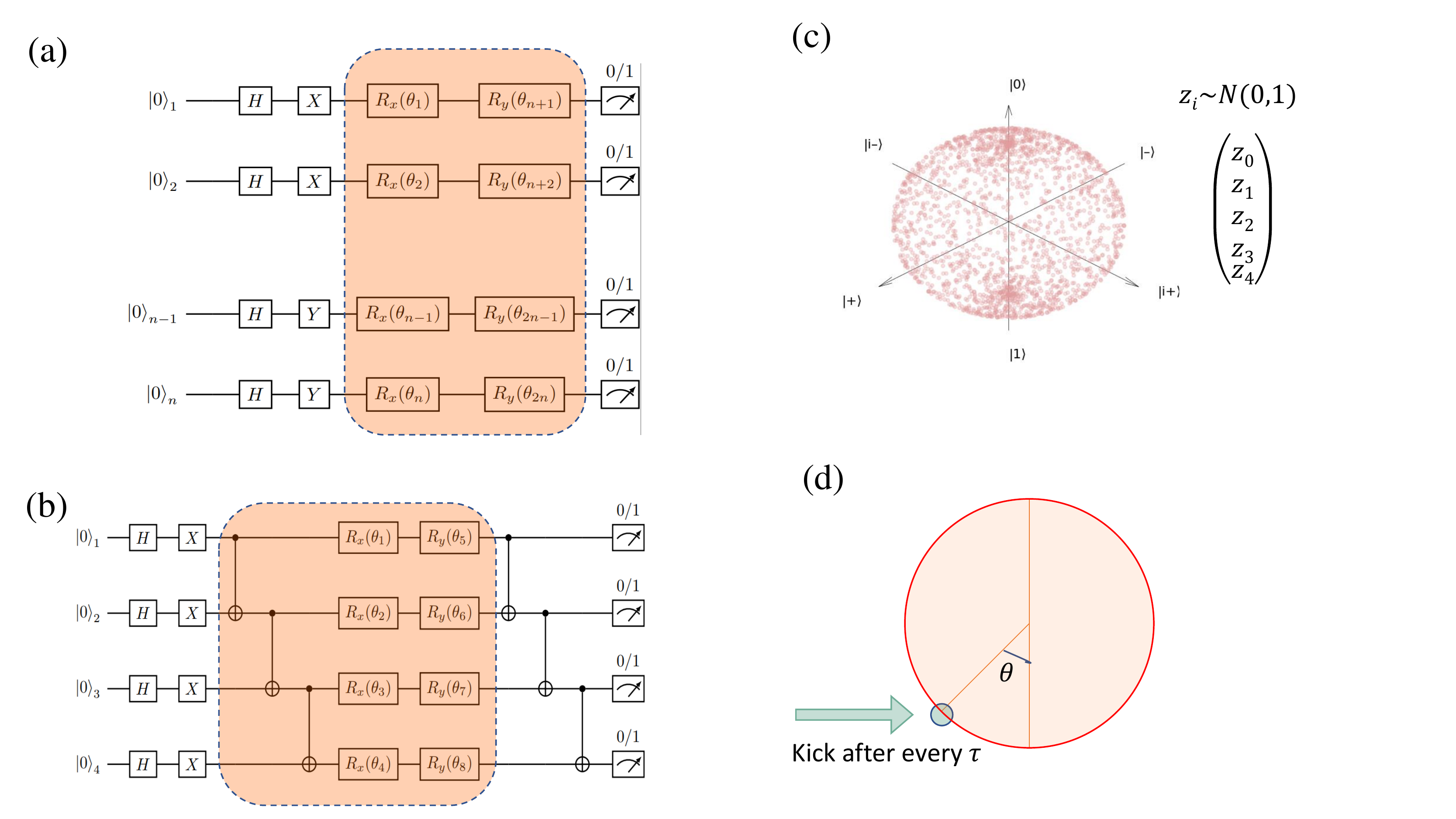}
\caption[Datasets obtained from measurement of different quantum states]{\textbf{Different types of measurement datasets}: (a) Product states obtained by a combination of arbitrary single qubit gates. The orange box represents repeatable layers of gates. (b) Quantum circuit states obtained from circuits with local (nearest-neighbor) entanglement. (c) Haar states obtained from a pure-quantum state by normalizing a $2^n$-dimensional complex vector. (d) Quantum-kicked rotor states obtained by the time-evolution of an initial state $|0\rangle$.}
\label{fig:generate-quantum-states}
\end{figure*}
\begin{figure*}[h!]
\centering
\includegraphics[width=0.9\textwidth]{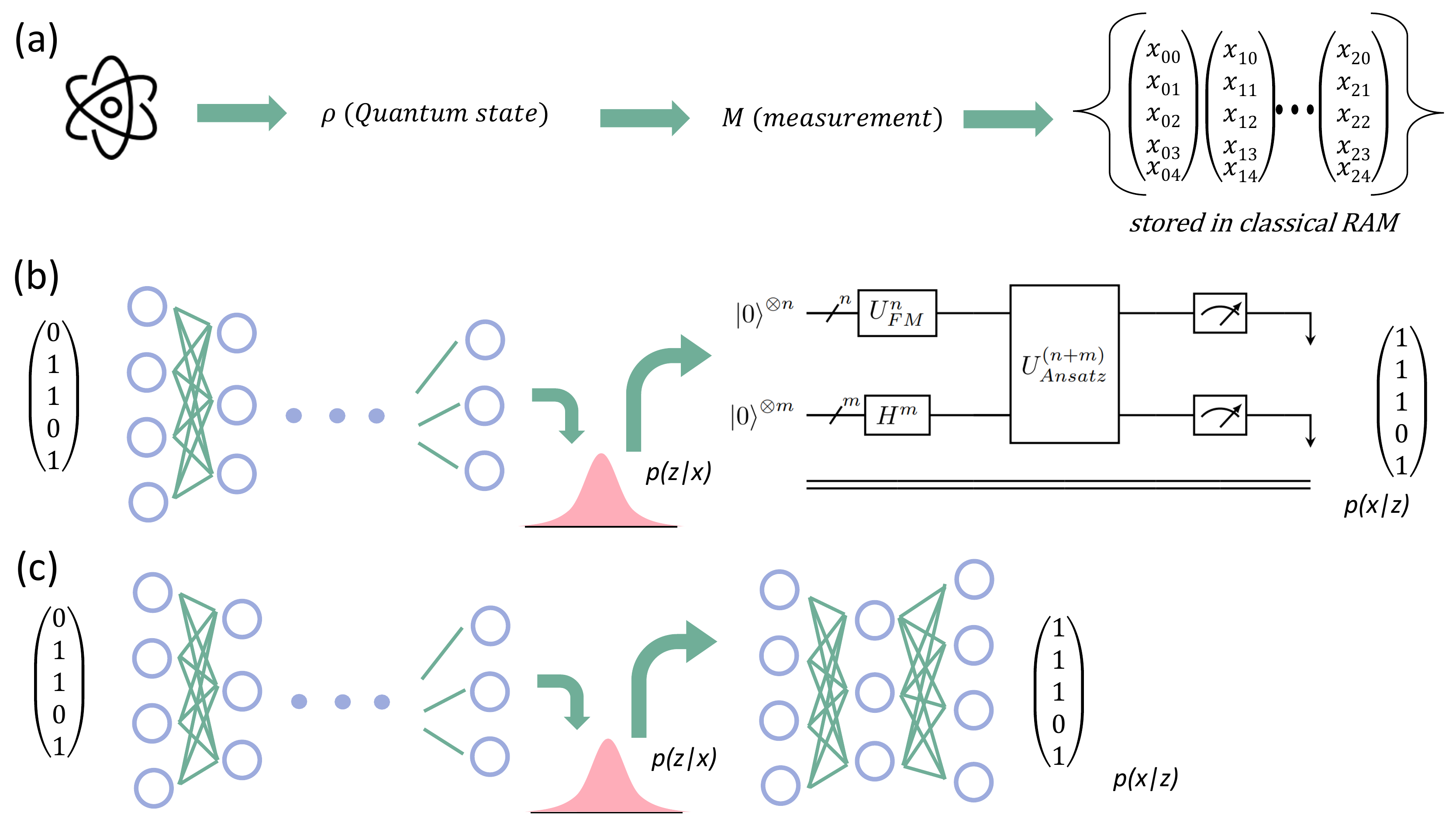}
\caption[QeVAE for learning quantum state distributions]{\textbf{QeVAE for learning quantum state distributions}: (a) Multiple copies of a quantum state $\rho$ are obtained naturally from a quantum system and are measured through different measurement operators. The measurement dataset is stored on a classical computer. (b) The QeVAE with a parameterized quantum circuit as the generative network and a classical feed-forward neural network as the inference network can be used to recrete the distribution. After training, the circuit can be used to generate the original distribution through any quantum computer and generate states amenable for downstream processing. (c) A classical VAE with continuous Gaussian latent variables that perform the same task.}
\label{fig:learn-quantum-state-method}
\end{figure*}
In order to evaluate our proposed approach, we try to model measurements from multiple families of states as described below. For each family we generate multiple 4 and 8 qubit states with different random seeds. Multiple copies of a state are measured in standard basis to yield a dataset consisting of samples in $\{0,1\}^n$ for each n qubit state. 
Here, each measurement dataset contains 1024 samples. We use 70\% samples for training and the  30\% for validation. We benchmark the performance of QeVAE on several datasets as shown in figure \ref{fig:generate-quantum-states}. We train both classical and quantum models to reproduce these distributions. We are interested in discovering if a classical learner can learn the same distribution and how the number of parameters required scales with the size of the system. In the following section, we describe the families of states that have been considered for creating our datasets. 

\emph{Random product states}: Product states (i.e. Tensor Product of single qubit states) are classically easy to simulate and are empirically found to be \emph{classically easy} to learn. We generate random product states by simulating quantum circuits with only single qubit gates with arbitrary angles of rotation, generated according to a random seed (figure \ref{fig:generate-quantum-states}(a)). The state prepared is on the form $|\psi\rangle = \otimes_{i=0}^{n-1} \left\{\alpha_i |0\rangle + \beta_i |1\rangle\right\}$, where $n$ is the number of qubits and projective Z-basis measurements generate the samples, $x \in \{0,1\}^n$. 

\emph{Haar random states}
are quantum states that are uniformly distributed over the Hilbert space according to the Haar measure and represent \emph{classically hard states} i.e, they require exponential number of parameters in the number of qubits to learn. These states can either be generated by first creating a Haar unitary $U$ and then applying it on an initial state of dimension $2^n$ or by normalizing a complex-valued vector of dimension $2^n$. We use the later method, where a complex-valued vector of $|\psi\rangle = \sum_{l=0}^{2^n-1}(c_{1l} + ic_{2l})|l\rangle$ is initialized with $|l\rangle$ corresponding to the orthonormal basis vector in the $2^n$-dimensional Hilbert space, $\mathbb{C}^{2^n}$, and $c_{1l},c_{2l}$ are real numbers chosen independently from a standard Gaussian distribution. This vector is normalized to yield a quantum state by using the constraint: $\langle \psi | \psi \rangle =1$. After normalization, the states are uniformly distributed on a unit hyper-sphere.

\emph{Random quantum circuit states} are obtained from random quantum circuits with a pre-defined entanglement structure and circuit depth, as shown in figure \ref{fig:learn-quantum-state-method}(b). These states are useful for circuit compression and circuit compilation. 

\emph{Quantum kicked rotor states} are obtained from the quantum kicked rotor (QKR), a well-known model in quantum chaos and quantum information that exhibits rich dynamics under time evolution. The QKR is defined by the Hamiltonian: 
\begin{align} 
	\hat{H} = \frac{\hat{p}^2}{2I} + k\cos{\hat{x}} \sum_n \delta(t - nT) 
\end{align} 
where $\hat{x}$ and $\hat{p}$ are the position and momentum operators, $I$ is the moment of inertia, $k$ is the kick strength, and $T$ is the kick period. We set $I=1$ and introduce dimensionless parameters $\hbar_s = \hbar T/I$ and $\kappa=k/\hbar$, where $\hbar$ is the reduced Planck's constant and $[\hat{x}, \hat{p}]=i\hbar_s$. The time-evolution operator for one period is then given by: 
\begin{align} 
	\hat{U}=\hat{U}_{kick}\hat{U}_{free}=\exp{\left(-i\kappa \cos{\hat{x}}\right)} \exp{\left(-\frac{i}{2\hbar_s} \hat{p}^2\right)} 
\end{align} 
To simulate the QKR, we apply $\hat{U}$ repeatedly to an initial state $|\psi_p(0)\rangle=0$ and perform forward and inverse Fourier transforms between each kick. The resulting wavefunction shows different behaviors depending on the value of $\kappa$. For weak kicking ($\kappa\lessapprox5.95$), the system exhibits quantum diffusion until a break time, where the variance of momentum $\langle p^2 \rangle$ grows linearly with time, and after the break time, $\langle p^2 \rangle$ saturates. For strong kicking ($\kappa\gtrapprox5.95$), the system exhibits dynamical localization, where $\langle p^2 \rangle$ reaches a finite value. In contrast, the classical kicked rotor shows chaotic behavior, where the future trajectory is highly sensitive to initial conditions. For a comprehensive review of the QKR, we refer the reader to \cite{SANTHANAM20221}. In this work, we investigate whether a generative model can learn the probability distribution of the wavefunction $|\psi_p|^2$ after 1000 kicks for different values of $\kappa \in \{0.5, 6 \}$ and $\hbar_s = 1$.


\section{Results}
\label{sec:results}

In the tables presented below, we summarize the best fidelity obtained across each type of measurement dataset. We compare the final fidelity between the target distribution and that produced by a random uniform guess, a classical variational autoencoder (CVAE), and a Quantum-enhanced variational autoencoder (QeVAE). For each type of state, we consider five different random seeds. QeVAE results include the best fidelity observed across different hyper-parameters like latent size, feature-map, prepossessing-layer, and relative KL-divergence term $\beta$.
\begin{table*}[h!]
\caption{Fidelity for Product states}
\centering
\begin{tabular}{ |l| *{6}{c} | *{6}{c}| }
	\hline
	No qubits & \multicolumn{6}{c|}{4} & \multicolumn{6}{c|}{8} \\
	\cline{1-13}
	Seed & 12 & 16 & 27 & 44 & 102 & \bf{Mean} $\downarrow$ & 12 & 16 & 27 & 44 & 102 & \bf{Mean} $\downarrow$ \\
	\hline
	Uniform & .471	& .356	& .302	& .156	& .436 & .344 & .164 & .306 & .081 &	.306 & .106 & .193\\
	CVAE & \textbf{.998} & \textbf{.995} & \textbf{.998} & \textbf{.995} & \textbf{.995} & \textbf{.996} & \textbf{.983} & \textbf{.979} & \textbf{.982} & \textbf{.983} &  \textbf{.987} & \textbf{.983} \\
	QeVAE & .995 & .827 & .882 & .973 & .892 & .914 & .875 & .947 & .870 & .823 & .957 & .894 \\
	\hline
\end{tabular}
\label{tab:product-states}
\end{table*}

\begin{table*}[h!]
\caption{Fidelity for Quantum circuit states}
\centering
\begin{tabular}{ |l| *{6}{c} | *{6}{c}| }
	\hline
	No qubits & \multicolumn{6}{c|}{4} & \multicolumn{6}{c|}{8} \\
	\cline{1-13}
	Seed & 12 & 16 & 27 & 44 & 102 & \bf{Mean} $\downarrow$ & 12 & 16 & 27 & 44 & 102 & \bf{Mean} $\downarrow$ \\
	\hline
	Uniform & .477 & .366 & .315 & .154 & .431 & .348 & .158 & .279 & .080 & .309 & .107 & .187 \\
	CVAE & .501 & .667 & .597 & \textbf{.925} & .758 & .690 & .229 & .423 & .079 & .304 & .181 & .243 \\
	QeVAE & \textbf{.981} & \textbf{.976} & \textbf{.950} & .873 & \textbf{.912} & \textbf{.938} & \textbf{.665} & \textbf{.654} & \textbf{.388} & \textbf{.548} & \textbf{.591} & \textbf{.569} \\
	\hline
\end{tabular}
\label{tab:qc-states}
\end{table*}

\begin{table*}[h!]
\caption{Fidelity for Haar random states}
\centering
\begin{tabular}{ |l| *{6}{c} | *{6}{c}| }
	\cline{1-13}
	No qubits & \multicolumn{6}{c|}{4} & \multicolumn{6}{c|}{8} \\
	\cline{1-13}
	Seed & 42 & 96 & 27 & 101 & 102 & \bf{Mean} $\downarrow$ & 12 & 43 & 16 & 27 & 2 & \bf{Mean} $\downarrow$ \\
	\hline
	Uniform & .772 & .776 & .777 & .771 & .768 & .773 & .766 & .770 & .773 & .781 & .772 & .772 \\
	CVAE & .795 & .798 & .800 & .788 & .788 & .794 & .754 & .755 & .757 & .766 & .763 & .759 \\
	QeVAE & \textbf{.839} & \textbf{.983} & \textbf{.913} & \textbf{.988} & \textbf{.932} & \textbf{.931} & \textbf{.876} & \textbf{.878} & \textbf{.887} & \textbf{.887} & \textbf{.887} & \textbf{.883} \\
	\hline
\end{tabular}
\label{tab:haar-states}
\end{table*}

\begin{table*}[h!]
\caption{Fidelity of Quantum-kicked rotor states}
\centering
\begin{tabular}{|l| ccc | ccc |}
	\cline{1-7}
	No qubits & \multicolumn{3}{c|}{4} & \multicolumn{3}{c|}{8} \\
	\cline{1-7}
	Type & Localized (k=6) & Diffusive (k=0.5)& \bf{Mean} $\downarrow$ & Localized (k=6) & Diffusive (k=0.5)& \bf{Mean} $\downarrow$ \\
	\hline
	Uniform & .175 & .838 & .506 & .053 & .418 & .236 \\
	CVAE & .723 & .908 & .815 & .061 & .406 & .233 \\
	QeVAE & \textbf{.991} & \textbf{.992} & \textbf{.991} & \textbf{.912} & \textbf{.616} & \textbf{.764} \\
	\hline
\end{tabular}
\label{tab:qkr-states}
\end{table*}

\begin{table}[h!]
\caption{Hardware results for a 4 qubit quantum circuit state}
\centering
\begin{tabular}{|l|c|c|c|c|c|c}
	\cline{1-6}
	State & Fidelity & Simulator & Hardware & Suppression & Mitigation \\
	\cline{1-6}
	Uniform & 0.477 & $\checkmark$ & & & \\
	CVAE & 0.501 & $\checkmark$ & & & \\
	QeVAE & 0.981 & $\checkmark$ & & & \\
	QeVAE & 0.658 & & $\checkmark$ &  & \\
	QeVAE & 0.642 & & $\checkmark$ & $\checkmark$ & $\checkmark$ \\
	\hline
\end{tabular}
\label{tab:hardware-results-qc12}
\end{table}

From tables(I-V), we observe that across all entangled quantum states we considered (i.e. all clases of states other than product states) the final fidelity obtained from a QeVAE outperforms the classical VAE and a random guess. In addition, the number of learnable parameters in the classical VAE is typically of the order $2^{(n)}$ (Table \ref{tab:cvae-parameters}) while those in a QeVAE is $4n+\epsilon$ where $n$ is the number of qubits and $\epsilon$ is a constant ($\epsilon <$ 4). To further validate our findings, we run the best QeVAE models on real quantum devices and see that the obtained fidelity is higher than those achieved by classical methods (Table \ref{tab:hardware-results-qc12}). We know from literature that the measurement distribution obtained from product states are classically \emph{easy} and that measurements obtained from states with quantum correlations is classically \emph{hard}. Our results not only corroborates with the above observation but also shows that quantum-enhanced classical models can overcome the drawbacks of purely classical models. Our main results and observations from data is as follows:

\begin{enumerate}
\item 
Our proposed algorithm achieves the highest fidelity across all types of datasets, other than product states. The inherent ability of our model to learn quantum correlations i.e we are able to produce entangled multi-qubit states through variational quantum circuits and tailor the rotation gates to reproduce a desired distribution allows them to outperform the classical model for the quantum circuit, haar random and kicked rotor states.
\item \emph{Fidelity for Product states}: Since product states do not employ quantum-correlations between different qubits in a many-body system, a distribution on $n$ bits obtained by measuring a product of $n$ qubits,  
can be characterized by $n$ independent Bernoulli variables, thus requiring only $n$ parameters.
This polynomial dependence and presence of only classical correlations (each bitstring in the distribution is independent) in the output distributions enables classical VAEs to efficiently learn and reproduce the distribution with very high fidelity. In addition, the ansatz in our model is not always tailored for product states, i.e the ansatz has an entanglement structure and the quantum state produced from the QeVAE will exhibit entanglement. Such a class of distribution produced cannot in general be similar to that produced by product states. This might explain the reduced fidelities for QVAEs. 
\item \emph{Fidelity for Haar, Quantum circuit and Quantum kicked rotor states}:
We notice that for measurement distributions obtained from more generic quantum states, our quantum models outperform the classical models in terms of final fidelity. In some cases (like seed 44 for quantum circuit states), the score of the classical model is higher because the resulting measurement distribution is highly concentrated around a single output ($>80\%$). In such cases, it becomes easier to just predict the output statistically, rather than learning the intrinsic structure of the quantum circuit producing the state. Nonetheless, we find that QeVAE models require only $O(n)$ rotation gates to learn the output distribution and achieve a fidelity score that is unattainable to other methods. 

\item Since QeVAE is a hybrid model, it operates through the synergy of quantum and classical resources. Classical models are well-versed at producing non-linear transformations whereas quantum models are restricted to unitary or linear transformations. Conversely, quantum models can encode quantum correlations like entanglement that inaccessible to classical methods. An important feature of our model is the ability to leverage the merits of both the classical and quantum models, and outperform both models. 

\item The QeVAE model contains many hyper-parameters that can be tuned to achieve optimal performance. In the above tables, we have presented the result for the best hyper-parameter setup. When the latent-size is set to zero, there is no contribution from the classical encoder and the output distribution is produced by the circuit alone. In such cases, the resulting model is the Quantum circuit Born Machine (QCBM), wherein the circuit parameters are iteratively updated to minimize the difference between the output and target distributions. Thus in the latent-size=0 limit, our model results in the QCBM generative model. The KL divergence term in the ELBO loss function can be neglected, and minimizing the negative expected log-likelihood becomes equivalent to minimizing the KL-divergence between the output and target distributions.

\end{enumerate} 

\textbf{Hardware run}:
We execute the best model trained on the simulator on IBM Hardware by transpiling the decoder circuit. To generate the output distribution, random samples from $N(0,1)$ are propagated through the preprocessor linear layer and then through the circuit (executed on hardware). An average over all initial random points yields the desired output distribution. Our results are depicted in figure \ref{fig:qcstate-hardware} and in table \ref{tab:hardware-results-qc12}. We find that final fidelity is lesser than that on a simulator. Using error mitigation and suppression techniques, we are able to perform better than the classical VAE.
\begin{figure}[h!]
\centering
\includegraphics[width=0.5\textwidth]{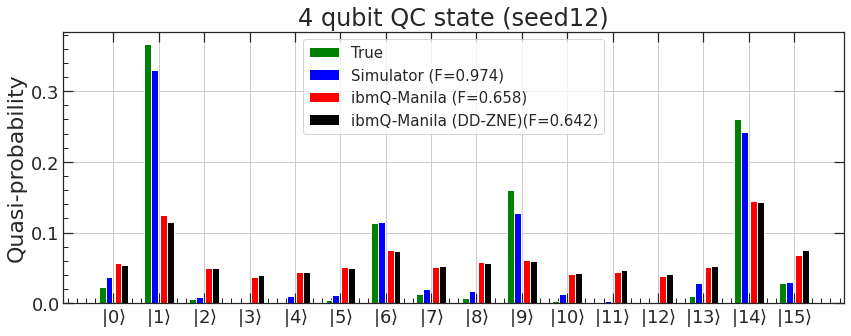}
\caption[QCstate hardware run]{\textbf{Inference on IBM hardware}: The true distribution produced by an unknown 4-qubit system (green) is used to train a QeVAE on IBM's qasm-simulator (blue) that results in a final of 0.97. After training, the decoder is executed on the IBM-Manila. The measurement distribution produced from hardware has a total fidelity of 0.658 which changes to 0.642 with error-mitigation (Zero noise extrapolation) and error-suppression (Dynamic decoupling)}
\label{fig:qcstate-hardware}
\end{figure}

\subsection{Circuit compilation}
\begin{figure*}[h!]
\centering
\includegraphics[width=0.95\textwidth]{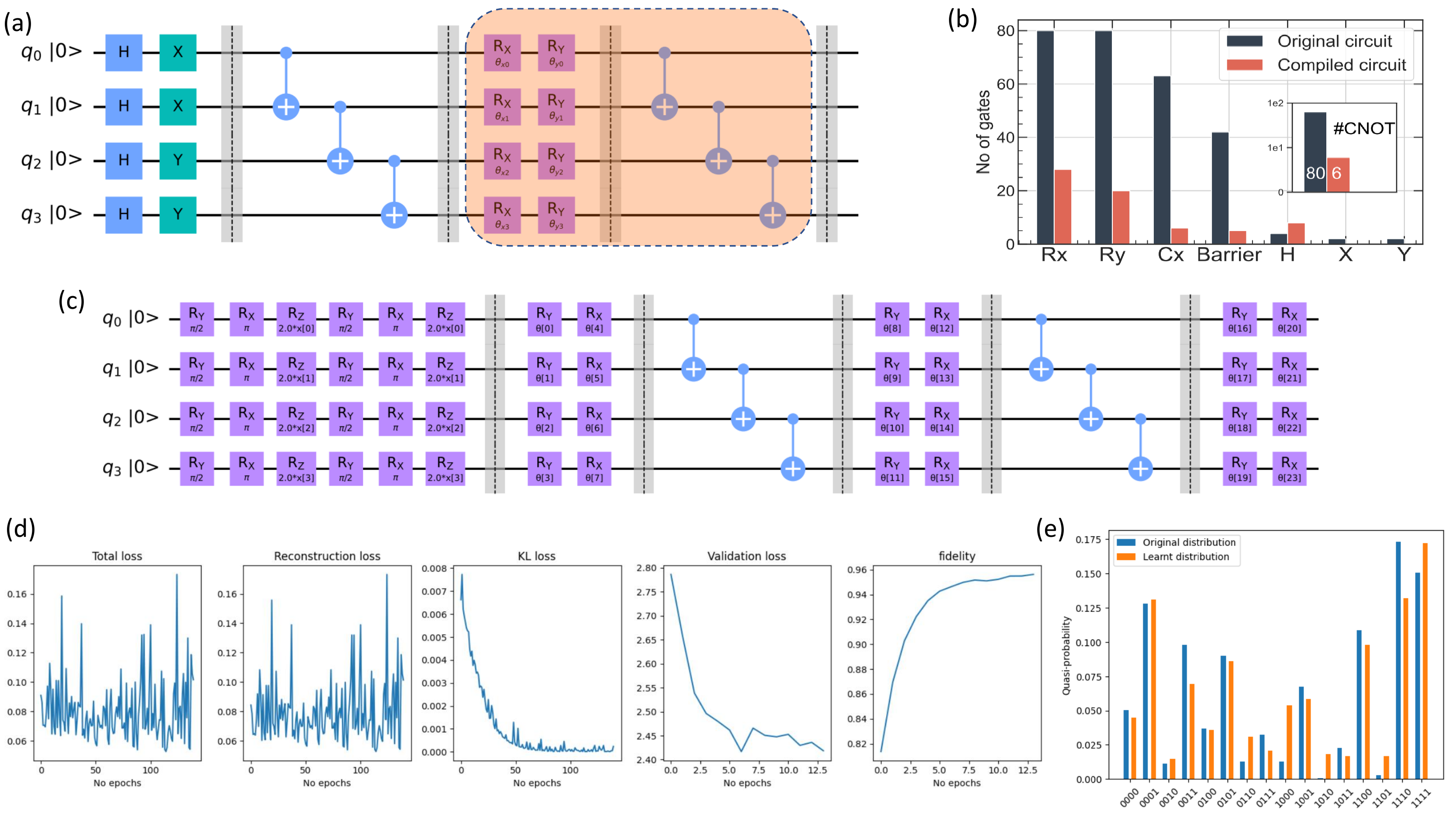}
\caption[QeVAE for circuit compilation]{\textbf{Circuit compilation with QeVAEs}: (a) Original circuit structure that produces a measurement distribution. The orange box represents a single layer of rotation and entangling layer that is repeated 20 times. (b) The compiled circuits requires very few gates when compared to the original circuit that produces the state. Particularly, the number of CNOT gates is reduced by $\sim$14X (c) The variational quantum circuit containing the Pauli-Z feature Map that embeds data x from the latent space and the ansatz with only 16 parameters. (d) Results of training the ansatz to produce the measurement distribution. We achieve a final fidelity of 0.956 (e) Original and the distribution produced by the QeVAE after training.}
\label{fig:circuit-compilation}
\end{figure*}
We now present a practical application of QeVAEs on the task of circuit compilation. Circuit compilation or circuit compression is an important area of focus in the NISQ era. Since the depth of circuits executable on hardware is limited, there is a need to transform deep circuits into shallow ones by altering the sequence of gates and reducing the overall size and complexity. Furthermore, multi-qubit
gates like the CNOT gate, T-gate, and SWAP gates are expensive to implement on hardware. Efficient circuit compilation assists in faster computations, and allows accurate simulation of quantum systems. This has immense applications in many-body physics, condensed matter physics, and quantum chemistry. Through simulation studies, we show show QeVAEs can help in reducing the complexity of circuits by learning to reproduce measurement distributions with fewer gates. 

We simulate an unknown quantum state by considering a deep quantum circuit with twenty layers of rotation and entangling gates. In reality, the form of the circuit is unknown and one only has access to the measurement data. A projective measurement on such a state produces a measurement distribution as shown in figure 4.5(e). Note that our goal here is to reproduce the measurement distribution and not to learn the original state itself. Through the QeVAE learning approach, we can learn the measurement distribution with high fidelity. After training, we can discard the encoder part of the circuit, and the decoder provides a sequence of gates that can be implemented on hardware to generate the same distribution. We achieve a final fidelity of 0.956
(figure \ref{fig:circuit-compilation}(d)) and a multi-fold reduction in the number of gates as seen in figure \ref{fig:circuit-compilation}(b). The initial circuit contains: 80 Rx gates, 80 Ry gates, 63 CNOT gates which are reduced to 28, 20, and 6 respectively, yielding an overall 4x reduction in total number of gates and 3x reduction in the number of parameters. 

\section{Conclusion and Outlook}
\label{sec:Conlcusion}
The major outcome of our work is a working hybrid quantum-classical machine learning model for generative learning and its bench-marking against purely classical and quantum models. In addition, our models are particularly suitable for quantum devices with small number of noisy qubits with limited connectivity. With this algorithm, we have investigated the ability of quantum generative models to learn the distributions obtained from the measurement of quantum many-body systems. Such distributions are known to be demanding for classical generative models and we verify the same in our experiments. We further go on to show that our proposed hybrid-model can learn these distributions with a much higher final fidelity. We find this trend to be universal across a range of different types of quantum states, from generic haar random states to dynamic kicked rotor states. We find that our model has multiple hyper-parameters to tune and in one such case, when the latent size is set to zero, we obtain at the Quantum circuit Born Machine, the most popular pure quantum-mechanical generative model. In addition, we have shown that QeVAEs can be useful for the practical task of circuit compilation.

Our work highlights the ability of quantum-enhanced models to perform better than classical models within the variational autoencoder framework. There are some drawbacks of the VAE approach and they include: (1) information loss in the encoding and decoding process that often results in blurry outputs; (2) posterior collapse: if the likelihood function is much more complex than the posterior, then the encoder ignores the input data and outputs a trivial latent space, leading the decoder to reconstruct the data from noise. Other classical generative methods like GANs, Diffusion Models, and Normalizing flows have gained much traction for producing better quality images. Although, they also contain their fair share of disadvantages like mode collapse in GANs and the large number of parameters in Diffusion models, they have been shown to produce better quality images. In future, quantized versions of such models can also be utilized to model measurement distributions.

The QeVAE algorithm contains many tunable parameters as mentioned in Table \ref{tab:hyperparameters}. In our calculations, we observed that obtaining an accurate final distribution requires careful fine-tuning of these hyperparameters. For example, we found out that initializing the preprocessing layer near $N(0,0.5)$ for our QeVAE is very essential to outperform the classical model for learning measurement distributions. In the absence of the preprocessing layer, the models quickly over-fit the training data and occasionally performed better than classical VAE. Moreover, we found that feature maps with entanglement (like the Pauli ZZ map) produced distributions with the same accuracy but require fewer epochs. The correct hyper-parameters have to be found out through an extensive grid-search, and through semi-empirical means.

In the future, work can focus on verifying our results for learning the classical distribution on diverse datasets. Particularly, text datasets can be used for discrete distributions. In addition, one can also examine if modifications in the ansatz like incorporating a non-linearity (mid-circuit measurements) can enhance the performance of the QeVAE and provide a better bound to the ELBO. Such hybrid quantum models can also be used to model the distribution of observables used in particle physics and condensed matter physics. 


\section{Acknowledgemnts}
ASR would like to thank the KVPY fellowship and IISER Pune computing resources for support. The corresponding code and datafiles are available on request.

\newpage
\bibliography{main}
\bibliographystyle{ieeetr}

\end{document}